\documentclass{article}

\usepackage{arxiv}

\usepackage[utf8]{inputenc} 
\usepackage[T1]{fontenc}    
\usepackage{hyperref}       
\usepackage{url}            
\usepackage{booktabs}       
\usepackage{nicefrac}       
\usepackage{microtype}      
\usepackage{lipsum}
\usepackage{graphicx}

\usepackage{amsmath,amsfonts,amsthm,amssymb,mathrsfs,dsfont,bbm} 
\usepackage{mathtools}
\usepackage{bbm}
\usepackage{hyperref}
\usepackage{subcaption}
\usepackage{xcolor}
\usepackage{algorithm}
\usepackage{algpseudocode}
\usepackage{natbib}

\newcommand{\cP}{\mathcal{P}}

\newcommand{\cS}{\mathcal{S}}

\newcommand{\cD}{\mathcal{D}}

\newcommand{\bP}{\mathbb{P}}
\newcommand{\bE}{\mathbb{E}}

\newcommand{\erf}{\text{erf}}

\newtheorem{theorem}{Theorem}[section]

\newtheorem{corollary}{Corollary}[section]
\newtheorem{lemma}[theorem]{Lemma}

\newtheorem{definition}[theorem]{Definition}

\newcommand{\E}{\bE}      

\newcommand{\argmax}{\operatornamewithlimits{argmax}}

\title{An Improved Upper Bound for the Euclidean TSP Constant Using Band Crossovers}

\author{
 Julia Gaudio \\
  Department of Industrial Engineering \& Management Sciences \\
  Northwestern University\\
  Evanston, IL 60208 \\
   \And
 Charlie K. Guan \\
  Department of Industrial Engineering \& Management Sciences \\
  Northwestern University\\
  Evanston, IL 60208 \\
}

\begin{document}
\maketitle
\begin{abstract}
Consider $n$ points generated uniformly at random in the unit square, and let $L_n$ be the length of their optimal traveling salesman tour. Beardwood, Halton, and Hammersley (1959) showed $L_n / \sqrt n \to \beta$ almost surely as $n\to \infty$ for some constant $\beta$. The exact value of $\beta$ is unknown but estimated to be approximately $0.71$ (Applegate, Bixby, Chv\'atal, Cook 2011). Beardwood et al. further showed that $0.625 \leq \beta \leq 0.92116.$ Currently, the best known bounds are $0.6277 \leq \beta \leq 0.90380$, due to Gaudio and Jaillet (2019) and Carlsson and Yu (2023), respectively. The upper bound was derived using a computer-aided approach that is amenable to lower bounds with improved computation speed. In this paper, we show via simulation and concentration analysis that future improvement of the $0.90380$ is limited to $\sim0.88$. Moreover, we provide an alternative tour-constructing heuristic that, via simulation, could potentially improve the upper bound to $\sim0.85$.
Our approach builds on a prior \emph{band-traversal} strategy, initially proposed by Beardwood et al. (1959) and subsequently refined by Carlsson and Yu (2023): divide the unit square into bands of height $\Theta(1/\sqrt{n})$, construct paths within each band, and then connect the paths to create a TSP tour. Our approach allows paths to cross bands, and takes advantage of pairs of points in adjacent bands which are close to each other. A rigorous numerical analysis improves the upper bound to $0.90367$. 
\end{abstract}


\section{Introduction}
For given vertices $x_1, \dots x_n$ on the unit square, the objective of the Euclidean traveling salesman problem (TSP) is to construct the shortest tour through all these vertices. The probabilistic analysis of \citet{BHH} yielded one of the first limit theorems in combinatorial optimization. 

\begin{theorem}[BHH Theorem] Generate vertices $X_1, \dots, X_n$ independently and uniformly on the unit square. Let $L_n$ denote the optimal TSP tour through $X_1, \dots, X_n$. There exists a constant $\beta$ such that 
\[ \lim_{n\to\infty} \frac{L_n}{\sqrt{n}} \to \beta \]
almost surely. The constant is bounded such that $0.625 \leq \beta \leq 0.92116.$
\end{theorem}
More generally, \cite{BHH} proved almost sure convergence results for vertices sampled independently from any common density function over a compact support of $\mathbb{R}^d$ for any integer $d\geq 2.$ The uniform sampling of vertices in $d=2$ is the simplest form of the ``probabilistic'' TSP to analyze. However, the value of $\beta$ is unknown, with the bulk of the literature focusing on numerically estimating $\beta$ using Monte Carlo estimates \citep{stein1977, ONG1989231, Fiechter1994APT, LeeChoi1994, Johnson1996AsymptoticEA, PercusMartin, VALENZUELA1997157, JacobsenReadSaleur}. The most recent Monte Carlo estimate by \citet{Applegate2006, concorde_tsp_solver} using very large simulated instances suggests that $\beta \approx 0.71.$ 

The original lower bound of $0.625$ by \citet{BHH} was shown by observing that, in a tour, every vertex is of degree two. As a result, the tour length can be lower-bounded by connecting each vertex to its two nearest neighbors. More recently, \citet{GAUDIO202067} improved the lower bound by refining the analysis to yield $0.6277.$
On the other hand, one may construct a heuristic that yields a tour and analyze its length to upper-bound $\beta$. The original upper bound of $0.92116$ by \citet{BHH} was shown by partitioning the unit square into horizontal bands and merging left-to-right traversing paths in each band. \citep{Steinerberger} improved the upper bound by $9/16 \times 10^{-6}$ by analyzing zig-zag structures where four consecutive vertices have a small horizontal difference but a large vertical difference and changing the visitation order. Most recently, \citet{YuCarlsson2023} generalized the approach by optimizing the visitation order for every $k$ consecutive vertices. Setting $k=4$ yielded an upper bound of $0.90380$, which is the best existing upper bound of $\beta$ to the best of our knowledge. 

In this paper, we focus on further upper-bounding $\beta$. First, we investigate how much more the approach of \citet{YuCarlsson2023} can improve $\beta$. Their approach, which we term the tuple optimization approach, is computer-aided and is thus amenable to future improvements with faster computation. We analyze this potential for improvement using Monte Carlo simulations. When $k=4$, the tuple optimization approach yields an upper bound of $\approx 0.885$. Increasing $k$ (at the expense of exponentially more computation), reduces the upper bound to $\approx 0.869$ when $k=8$. 

We then provide a new heuristic that uniformly improves the upper bound of \citet{YuCarlsson2023}. The key idea is to move beyond the band constraint, which limits the heuristic to construct paths that only traverse through vertices within each band. 
Our approach allows vertices near the band boundary to be appended to a path in the adjacent band. Using Monte Carlo and analyzing its concentration, we estimate this upper bound to be $\approx 0.869$ for $k=4$ and $\approx 0.850$ for $k=8.$ Furthermore, we rigorously improve the upper bound such that $\beta \leq 0.90367$, improving the bound of \cite{YuCarlsson2023} by $0.00013$. 

The rest of the paper is structured as follows. Section \ref{sec:prev} analyzes the previous approaches to upper-bound $\beta$ in greater detail. Section \ref{sec:tuple_improve} analyzes the potential for improvement of the tuple approach using Monte Carlo. Section \ref{sec:band_crossing_approach} introduces our band crossing strategy and provides an improved valid upper bound for $\beta$.

\section{Previous approaches for the upper bound}
\label{sec:prev}

First, we review the original argument by \citet{BHH} that showed $\beta \leq 0.92116$. One may replace the $n$ independent points by a Poisson process with intensity $n$ and characterize the limiting tour length in expectation, as justified by the following lemma. 
\begin{lemma}[\citet{BHH}]
Let $\cP_n$ denote a Poisson process with rate $n$ on $[0, 1]^2$. It holds that $\E[L(\cP_n)] / \sqrt{n} \to \beta$ almost surely.
\end{lemma}

Next, the unit square is partitioned into horizontal bands of height $h/ \sqrt{n}$, where $h$ is a constant to be chosen later. Within each band, the approach connects the vertices by traversing them from left to right. A tour is finally constructed by adding edges to connect paths in adjacent bands. The total length needed to aggregate all band-paths is almost surely $o(\sqrt{n})$. As a result, the only component that contributes to $\beta$ is the expected total length of the paths within the bands. 

Consider any path $(X_1, Y_1), \dots, (X_k, Y_k)$, where we order $X_1 \leq \dots \leq X_k$. Since the vertices are generated from a Poisson process, it holds that $Y_i \sim h U_i / \sqrt{n}$, where the $U_i$'s are independent uniform random variables on $[0, 1]$, and the difference between consecutive $X_i$'s follows a scaled exponential distribution; i.e., $h\sqrt{n}(X_{i+1} - X_i) \sim Z_i$, where the $Z_i$'s are independent exponential random variables with mean $1$. It follows that the expected distance between consecutive vertices is 
\begin{align*}
    \E \bigg\|\begin{pmatrix}
        X_{i+1} - X_i  \\ Y_{i+1} - Y_i  
    \end{pmatrix} \bigg\| 
    &= 
    \frac{1}{h\sqrt{n}}
    \E \bigg\| \begin{pmatrix}
        Z_i  \\ h^2(U_{i+1} - U_i)  
    \end{pmatrix} \bigg\|.
\end{align*}
We upper-bound the total expected length by multiplying by $n$ and then upper-bound $\beta$ by dividing by $\sqrt{n}$ to obtain 
\begin{align}
    \beta &\leq \frac{1}{h}
    \E \bigg\| \begin{pmatrix}
        Z_i  \\ U_{i+1} - U_i  
    \end{pmatrix} \bigg\| \nonumber \\
    &= \frac{1}{h} \int_0^\infty \int_0^1 \int_0^1 e^{-z} \sqrt{z^2 + h^4(u_0 - u_1)^2} du_0 du_1 dz \nonumber \\
    &= \frac{1}{3h^5}\int_0^\infty e^{-z}\left( 3h^2z^2\log\left(h^2/z + \sqrt{h^4/z^2+1}\right) + 2z^3 + (h^4-2z^2)\sqrt{h^4+z^2}\right)dz. \label{eq:bhh}
\end{align}
After evaluating \eqref{eq:bhh} under different values of $h$ using Simpson's rule, \citet{BHH} reported the lowest upper bound of $0.92037$ at $h=\sqrt{3}$. However, this bound was corrected to $0.92116$ by \citet{Steinerberger} after fixing a numerical error. In addition, \citet{Steinerberger} improved the upper bound by $9/16 \times 10^{-6}$ by observing that the left-to-right traversal is not always optimal as it may induce unnecessary zigzag structures within the band. Under specific configurations when a set of four consecutive vertices have a small horizontal difference but a large vertical difference, an alternate vertex visitation order is preferable and reduces the path length. The numerical improvement is calculated by lower-bounding the probability of these zigzag configurations from occurring and analyzing the improvement in the path length. 

Most recently, \citet{YuCarlsson2023} improved the upper bound by further refining the vertex visitation order. Under the same band partition of height $h/\sqrt{n}$, fix a parameter $k$ and consider the set of ordered $(k+1)$-tuples $(X_i, Y_i), \dots, (X_{i+k}, Y_{i+k})$ such that $X_i \leq \dots \leq X_{i+k}$. 
Then, their approach selected the permutation that minimizes the total path length traversing these $k+1$ points. The permutations are restricted to those that begin with the leftmost endpoint $(X_i, Y_i)$ and end with the rightmost endpoint $(X_{i+k}, Y_{i+k})$ in order to yield segments that can be aggregated to form a valid tour.
For a given $k$, the tuple contains $k+1$ vertices since the right endpoint of one tuple becomes the left endpoint of the next tuple in the band. 
The base case of $k=2$ reduces the analysis to the original approach of \citet{BHH}.

Denote $\widehat \beta_k$ as the upper bound obtained using this tuple optimization approach. For illustration, consider $k=3$. Since there are at most $n/3$ tuples, $\beta$ is upper-bounded by 
\begin{align*}
    \beta \leq \widehat \beta_3 &:= \frac{1}{3h} \E \min\!\left\{
\begin{aligned}
  &\left\lVert
      \begin{pmatrix}
        Z_{1}\\[2pt]
        h^{2}\!\bigl(U_{0}-U_{1}\bigr)
      \end{pmatrix}
    \right\rVert
   +\!
    \left\lVert
      \begin{pmatrix}
        Z_{2}\\[2pt]
        h^{2}\!\bigl(U_{1}-U_{2}\bigr)
      \end{pmatrix}
    \right\rVert
   +\!
    \left\lVert
      \begin{pmatrix}
        Z_{3}\\[2pt]
        h^{2}\!\bigl(U_{2}-U_{3}\bigr)
      \end{pmatrix}
    \right\rVert\\
  &\left\lVert
      \begin{pmatrix}
        Z_{1}+Z_{2}\\[2pt]
        h^{2}\!\bigl(U_{0}-U_{2}\bigr)
      \end{pmatrix}
    \right\rVert
   +\!
    \left\lVert
      \begin{pmatrix}
        Z_{2}\\[2pt]
        h^{2}\!\bigl(U_{2}-U_{1}\bigr)
      \end{pmatrix}
    \right\rVert
   +\!
    \left\lVert
      \begin{pmatrix}
        Z_{3}\\[2pt]
        h^{2}\!\bigl(U_{1}-U_{3}\bigr)
      \end{pmatrix}
    \right\rVert
\end{aligned}
\right\}.
\end{align*}
For general $k$, by denoting $X_i = \sum_{j=1}^i Z_j$ and $\Pi_k$ as the set of all permutations $\pi$ of $\{0, \dots, k\}$ such that $\pi(0)=0$ and $\pi(k) = k$, we obtain 
\begin{align}
    \widehat \beta_k &:= \frac{1}{kh}\E \min_{\pi\in\Pi_k}\sum_{i=1}^k \bigg\| \begin{pmatrix}
        X_{\pi(i)} -  X_{\pi(i-1)} \\ h^2(U_{\pi(i)} -  U_{\pi(i-1)}) 
    \end{pmatrix} \bigg\| \label{eq:yu_carlsson_mc} \\
    &= \frac{1}{kh}\int_{\cD} e^{-x_k} \min_{\pi\in\Pi_k}\sum_{i=1}^k \bigg\| \begin{pmatrix}
        x_{\pi(i)} -  x_{\pi(i-1)} \\ h^2(u_{\pi(i)} -  u_{\pi(i-1)}) 
    \end{pmatrix} \bigg\| dV, \label{eq:yu_carlsson_bound}
\end{align}
where $\cD = \{(x_1, \dots x_k), (u_0, \dots, u_k): 0\leq x_1 \leq \dots \leq x_k, 0\leq u_i \leq 1\}$ is the domain of integration.
\citet{YuCarlsson2023} bounded the integral of \eqref{eq:yu_carlsson_bound} using a simulation-based and computer-aided approach with $k=4$ to yield a valid upper bound of $\widehat \beta_4 \leq 0.90380$. 
Specifically, the integral over $\cD$ of each argument inside the minimum operator in \eqref{eq:yu_carlsson_bound} can be bounded from above since each term can either be analytically evaluated or, due to convexity, upper-bounded by a product of secant lines, which is then numerically integrated via the trapezoidal rule.
To obtain a good bound, one must partition $\cD$ into granular, disjoint boxes and, in each box, select the optimal permutation that minimizes integral over that subregion. 
The authors obtained such partition by simulating $N_{\mathrm{sim}}$ samples of $(k+1)$ tuples uniformly at random, calculating the optimal permutation for each sample, and training a decision tree, whose leaves naturally induce a partition of the domain, to predict the optimal permutation. 
Their best bound of $\widehat \beta_4 \leq 0.90380$ was obtained by simulating $N_{\mathrm{sim}}=3.6\times 10^8$ and training the tree until the the remaining vertices at every leaf node exhibited uniform optimal permutation. While greater $k$ and $N_{\mathrm{sim}}$ could potentially lower this bound, the computation is intensive. Nevertheless, their approach is amenable to further improvements over time as computer hardware improves. 

\section{How much more improvement could we obtain with tuple optimization?}
\label{sec:tuple_improve}

A natural question to ask is  how much the tuple optimization approach could improve in the future. We estimate the limit using Monte Carlo simulations and provide an accompanying concentration analysis to show with a high degree of confidence that $\widehat \beta_4 \approx 0.88$, suggesting the potential for improvement is approximately $0.02$. 

Denote $\widehat \beta_{k, M}$ as the Monte Carlo estimate for $\widehat \beta_k$ using $M$ replicates, defined via
\[
\widehat \beta_{k,M}:=\frac{1}{Mkh}\sum_{i=1}^M L_i,
\qquad
\widehat \beta_k := \frac{\E L}{kh},
\]
where $L_1,\dots,L_M$ are i.i.d.\ copies of
\[
L(X_0,\dots,X_k,U_0,\dots,U_k)
=\min_{\pi\in\Pi_k}\sum_{i=1}^k 
\left\|
\begin{pmatrix}
X_{\pi(i)}-X_{\pi(i-1)}\\
h^2(U_{\pi(i)}-U_{\pi(i-1)})
\end{pmatrix}
\right\|.
\]
 For each replicate, we simulate $k+1$ vertices uniformly at random in the band and compute the minimum path length under the optimal permutation in \eqref{eq:yu_carlsson_mc}.
Table \ref{tab:mc_k4} shows the results for $\widehat \beta_{4, 10^7}$. Across all $h$ that were tested, the bound is approximately $0.88-0.89$, which suggests the bound of $0.90380$ by \cite{YuCarlsson2023} will be unlikely to improve beyond it, since the rigorous upper bound contains numerical approximations that inflate the bound. 

\begin{table}[t]
  \centering
  \caption{Simulation results for the tuple optimization approach for $k=4$ and varying $h$ using $10^7$ replicates.}
  \label{tab:mc_k4}
  \begin{tabular}{ccc}
    \hline
    $h^2$ & $\widehat \beta_{4, 10^7}$ & Standard error \\
    \hline
    3.00 & 0.890205 & 0.000091 \\
    3.25 & 0.886547 & 0.000089 \\
    3.50 & 0.884783 & 0.000087 \\
    3.75 & 0.884487 & 0.000086 \\
    4.00 & 0.884916 & 0.000085 \\
    \hline
  \end{tabular}
\end{table}

Next, we investigate $\widehat\beta_{k, M}$ for higher values of $k$. We fix $h^2=3.75$, which achieved the lowest bound at $k=4$. Table \ref{tab:mc_yc_vary_k} shows increasing $k$ lowers the bound, achieving the best bound of $0.862225$ when $k=8$.

\begin{table}[t]
  \centering
  \caption{Simulation results for the tuple optimization approach for $h^2=3.75$ and varying $k$ using $10^7$ replicates.}
  \label{tab:mc_yc_vary_k}
  \begin{tabular}{ccc}
    \hline
    $k$ & $\widehat \beta_{k, 10^7}$ & Standard error \\
    \hline
    5 & 0.874843 & 0.000077 \\
    6 & 0.869072 & 0.000071 \\
    7 & 0.865011 & 0.000066 \\
    8 & 0.862225 & 0.000062 \\
    \hline
  \end{tabular}
\end{table}

Using concentration analysis, we show that these Monte Carlo estimates at $M=10^7$ tightly concentrate around $\widehat \beta_k$.

\begin{theorem}
\label{thm:conc_beta_chernoff_correct}
Fix any $k\in\mathbb N$, $M\in\mathbb N$, and $h>0$. 
Then for any $\delta\in(0,1)$,
\[
\bP\!\left(\big|\widehat\beta_{k,M}-\widehat\beta_k\big|\ge \varepsilon\right)\le \delta,
\]
where with $t=\log(4/\delta)$ one valid choice is
\begin{equation}
\label{eq:eps_ch_correct}
\varepsilon
=
\frac{1}{kh}
\left(\sqrt{\frac{2kt}{M}} + \frac{t}{M}
+
kh^2\sqrt{\frac{t}{2M}}
\right). 
\end{equation}
\end{theorem}

Apply Theorem~\ref{thm:conc_beta_chernoff_correct} with $k=4$, $M=10^7$, $h^2=3.75$, and $\delta=0.001$.
Then $t=\log(4/\delta)=\log(4000)$ and
\[
\varepsilon=1.579704\hdots\times 10^{-3}.
\]
Thus, using the estimate $\widehat\beta_{4,10^7}=0.884487$ from Table \ref{tab:mc_k4}, we obtain with probability at least $0.999$ that
\[
\beta_4 \in \widehat\beta_{4,10^7}\pm 1.579704\times 10^{-3}
=
[0.88290730,\;0.88606670].
\]

\begin{proof}[Proof of Theorem~\ref{thm:conc_beta_chernoff_correct}]

Let $G=X_k - X_0$ denote the horizontal span of the tuple, which is the sum of $k$ independent
$\mathrm{Exp}(1)$ random variables. Thus $G\sim \mathrm{Gamma}(k,1)$.
Consider the path induced by the identity permutation, which stochastically dominates $L$
since $L$ is a minimum over all valid permutations. The sum of horizontal lengths is exactly $G$.
Since $U_i\sim \mathrm{Unif}(0,1)$, the vertical increment between consecutive vertices satisfies
\[
\big|h^2(U_{\pi(i)}-U_{\pi(i-1)})\big|\le h^2,
\]
and therefore each step length is at most the horizontal increment plus $h^2$.
Summing over $k$ steps yields $L \le G + kh^2$.
On the other hand, any path from the left endpoint to the right endpoint must cover the horizontal span,
so $G\le L$. Hence we have the sandwich relation
\[
G \le L \le G+kh^2,
\]
equivalently,
\[
L = G+B,\qquad 0\le B \le kh^2,
\]
for some bounded remainder random variable $B$.

Let $L_i=G_i+B_i$ be i.i.d.\ copies with sample means
\[
\bar L:=\frac1M\sum_{i=1}^M L_i,\qquad
\bar G:=\frac1M\sum_{i=1}^M G_i,\qquad
\bar B:=\frac1M\sum_{i=1}^M B_i.
\]
Then $\bar L-\E L=(\bar G-\E G)+(\bar B-\E B)$, so for any $a_G,a_B>0$,
\begin{equation}
\label{eq:union_GB}
\bP\big(|\bar L-\E L|\ge a_G+a_B\big)
\le 
\bP\big(|\bar G-k|\ge a_G\big)
+
\bP\big(|\bar B-\E B|\ge a_B\big).
\end{equation}

\paragraph{Step 1: Chernoff for the Gamma mean $\bar G$}

Let $S:=\sum_{i=1}^M G_i$. Since sums of independent Gammas are Gamma,
\[
S\sim \mathrm{Gamma}(kM,1),\qquad \E S = kM,\qquad \bar G=S/M.
\]

\begin{definition}
\label{def:subgamma}
A real-valued centered random variable $X$ is said to be sub-gamma with variance factor $v$
and scale parameter $c$ if
\[
\log \E e^{\lambda X} \leq \frac{\lambda^2v}{2(1-c\lambda)}
\]
for every $\lambda$ such that $0 < \lambda < 1/c$.
\end{definition}

\begin{lemma}[\cite{concentration_inequalities}, Section 2.4]
\label{thm:subgamma_concentration}
Let $X$ be a sub-gamma random variable with parameters $v$ and $c$. Then, for $t>0$,
\[
\bP\!\left( |X| \geq \sqrt{2vt}+ct\right) \leq 2e^{-t}.
\]
Moreover, if $Y\sim \mathrm{Gamma}(a, 1/b)$, then $X=Y-\E Y$ is sub-gamma with variance factor $ab^2$ and scale parameter $b$.
\end{lemma}

By Lemma \ref{thm:subgamma_concentration} and Definition \ref{def:subgamma}, we have that $\bar G - \E \bar G$ is sub-gamma with variance factor $k/M$ and scale parameter $1/M$. The concentration result yields
\begin{equation}
\label{eq:G_two_sided_correct}
\bP\!\left(|\bar G-k|\ge a_G(t)\right)\le 2e^{-t},
\end{equation}
where 
\[ a_G(t) = \sqrt{\frac{2kt}{M}} + \frac{t}{M}.\]

\paragraph{Step 2: Hoeffding for the bounded mean $\bar B$}
Since $0\le B_i\le kh^2$ a.s., Hoeffding's inequality \citep[Lemma 2.2]{concentration_inequalities} yields for all $t>0$,
\begin{equation}
\label{eq:B_two_sided}
\bP\!\left(|\bar B-\E B|\ge a_B(t)\right)\le 2e^{-t},
\qquad
a_B(t):=kh^2\sqrt{\frac{t}{2M}}.
\end{equation}

\paragraph{Step 3: Combine and normalize}
Combining \eqref{eq:union_GB}, \eqref{eq:G_two_sided_correct}, and \eqref{eq:B_two_sided} gives
\[
\bP\!\left(|\bar L-\E L|\ge a_G(t)+a_B(t)\right)\le 4e^{-t}.
\]
Choose $t=\log(4/\delta)$ so the right-hand side equals $\delta$.
Finally, since $\widehat\beta_{k,M}-\widehat \beta_k=(\bar L-\E L)/(kh)$, dividing the deviation by $kh$
yields \eqref{eq:eps_ch_correct}.
\end{proof}

\section{Improving the upper bound using band crossings}
\label{sec:band_crossing_approach}
In this section, we present a tour-constructing heuristic that improves upon the upper bound of the tuple optimization approach given by \citet{YuCarlsson2023} for any $k$ and $h$. Thus, any future reduction of the upper bound from the tuple optimization approach automatically translates to a further reduced upper bound using our strategy. We consider the same band partition of height $h/\sqrt{n}$ and $(k+1)$-tuples, but the key idea of our improvement procedure is to allow appending a vertex to another vertex in the adjacent band if they are close to the band boundary. This is in contrast to all previous approaches that constrain paths to traverse through vertices within one band at a time. 

As a concrete example, consider $k=3$, as shown in Figure \ref{fig:k3}. The original approach of \citet{BHH} traverses the vertices in a left-to-right order, as indicated by the brown path. The approach of \cite{YuCarlsson2023} also considers the non-identity permutation for visiting the tuple, as indicated by the blue path. 
Our heuristic allows for a third option indicated by the pink path, where the vertices $u$ and $v$ that are close to each other across the band boundary is connected by a $2$-cycle, and the remaining vertices in the tuple are connected by a separate path. The $2$-cycle does not induce a valid TSP tour, but it upper-bounds the optimal TSP tour since a tour can be constructed by the following modification: remove one of the $(u, v)$ edges, remove the edge $(v, w)$, and insert the edge $(u, w)$.  
\begin{figure}[t]
\centering
\includegraphics[width=0.9\linewidth]{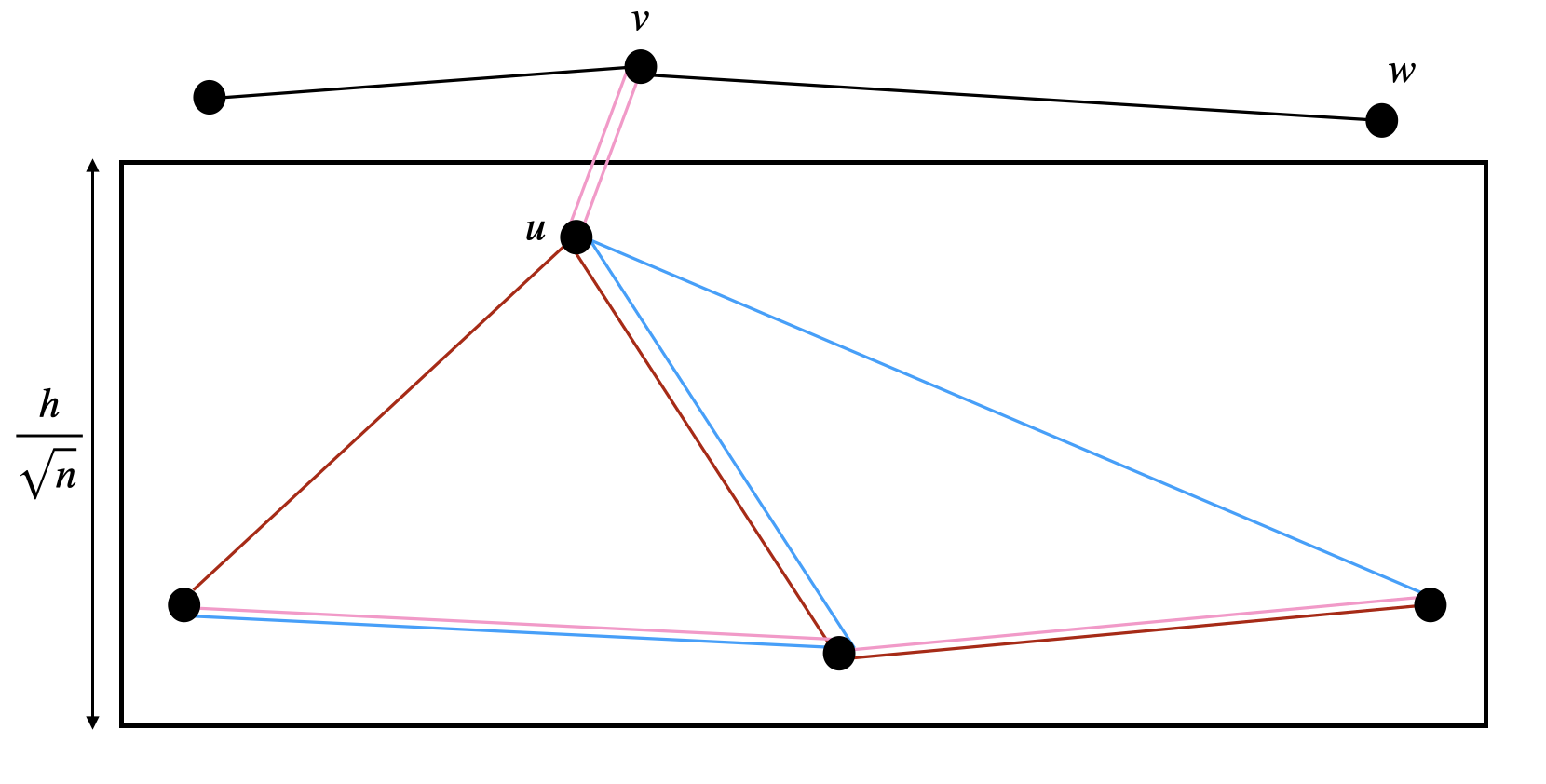}
\caption{By allowing band crossovers, there exist $3$ options to traverse the vertices in a $(k+1)$-tuple when $k=3$. The brown path is induced by the identity permutation $1234$, the blue path is induced by the permutation $1324$, and the pink path traverses the first, third, and fourth vertices in the tuple while appending vertex $u$ to the adjacent band.}\label{fig:k3}
\end{figure}
Due to the triangle inequality,
\[ \| (u, w) \| + \| (u, v) \| \leq 2\| (u, v) \| + \| (v, w) \|. \]
We apply this crossing strategy starting from the bottom band of the square and work upwards until the last band. For every tuple, we opt for the band crossing when its length beats any of the permutations traversing through the tuple. As a result, our procedure is guaranteed to improve $\widehat \beta_k$ for any $k$ and $h$.

\subsection{Simulation results}
We denote $\widetilde \beta_k$ as the upper bound on $\beta$ derived using the band crossover strategy for $k$. First, we estimate $\widetilde \beta_k$ via its corresponding Monte Carlo estimator $\widetilde \beta_{k, M}$, where $M$ is the number of replicates.   
Denote 
\[ J^\star~=~\argmax_{j\in\{1, \dots, k-1\}} U_j \]
as the highest interior vertex of the tuple and denote $\Xi_k$ as the set of all permutations $\xi$ of $\{0, \dots, k\} \setminus J^\star$ such that $\xi(0)=0$ and $\xi(k) = k$. Then,

\begin{align}
    \widetilde \beta_k &\leq \frac{1}{kh}\E\Bigg[  \min\Bigg\{ \min_{\pi\in\Pi_k}\sum_{i=1}^k \Bigg\| \begin{pmatrix}
        X_{\pi(i)} -  X_{\pi(i-1)} \\ h^2(U_{\pi(i)} -  U_{\pi(i-1)}) 
    \end{pmatrix} \Bigg\|,  \nonumber \\
    &\quad  2h \Bigg\| \begin{pmatrix}
        R\cos\theta \\  h(1-U_{J^\star}) + R\sin \theta
    \end{pmatrix} \Bigg\| + \min_{\xi\in\Xi_k}\sum_{i=1}^{k-1} \Bigg\| \begin{pmatrix}
        X_{\pi(i)} -  X_{\pi(i-1)} \\ h^2(U_{\pi(i)} -  U_{\pi(i-1)}) 
    \end{pmatrix} \Bigg\| \Bigg\}  \Bigg], \label{eq:band_exp}
\end{align}
where $R$ is the distance of the nearest vertex in the adjacent band to the projection of the $J^\star$-th vertex to the bound boundary, and $\theta$ is the angle of such neighbor relative to the boundary, uniformly distributed on $[-\pi/2, \pi/2]$. Conditioned on $J^\star$, the event $R/\sqrt{n} \geq t / \sqrt{n}$ for any $t>0$ is equivalent to the event that there exist no vertices in a half-circle of radius $t / \sqrt{n}$. The number of vertices in such half-circle is distributed as a Poisson random variable with parameter $\pi t^2/2$. Hence, $\bP(R \geq t) = \exp(-\pi t^2 /2)$, i.e., $R$ is distributed as a Rayleigh random variable with scale parameter $1/\sqrt\pi.$
The first inner minimum in \eqref{eq:band_exp} is the optimal permutation to visit all $(k+1)$ vertices within band using the tuple optimization approach.
The first term in the second argument of the outer minimum in \eqref{eq:band_exp} is the length of the 2-cycle using the band crossover approach, and the second inner minimum is the optimal permutation to visit the remaining vertices in the tuple.

Denote the expression inside the expectation in \eqref{eq:band_exp} as $\widetilde L$ and $\widetilde L_i$ as its independent copy for the $i$-th replicate. 
Then, $\widetilde \beta_k$ can estimated using Monte Carlo via
\begin{align}
    \widetilde \beta_{k, M} &\leq \frac{1}{Mkh}\sum_{\ell=1}^{M}\widetilde L_i. \nonumber
\end{align}
Table \ref{tab:mc_band_k4} shows the numerical results for $k=4$. The band crossover strategy yields an improvement of $0.011$ to $0.019$ across varying $h$ over the tuple optimization approach. The best bound of $0.868131$ was achieved at $h^2=4.00.$
\begin{table}[htbp]
  \centering
  \caption{Simulation results for the band crossover approach for $k=4$ and varying $h$ using $10^7$ replicates.}
  \label{tab:mc_band_k4}
  \begin{tabular}{ccccc}
    \hline
    $h^2$ & $\widetilde \beta_{k, M}$ & Std. error (of $\widetilde \beta_{k, M}$) & $\widehat \beta_{k, M}$ & Improvement\\
    \hline
    3.00 & 0.878662 & 0.000090 & 0.890205 & 0.011543 \\
    3.25 & 0.873754 & 0.000088 & 0.886547 & 0.012793 \\
    3.50 & 0.870653 & 0.000086 & 0.884783 & 0.014130 \\
    3.75 & 0.868938 & 0.000085 & 0.884487 & 0.015549 \\
    4.00 & 0.868131 & 0.000084 & 0.884916 & 0.016785 \\
    4.25 & 0.868511 & 0.000083 & 0.886460 & 0.017949 \\
    4.50 & 0.869483 & 0.000082 & 0.888705 & 0.019222 \\
    \hline
  \end{tabular}
\end{table}
Fixing $h^2=4.00$, increasing $k$ lowers the upper bound further. As indicated in Table \ref{tab:mc_band_vary_k}, we obtain the best bound of $0.849514$ at $k=8$.
\begin{table}[htbp]
  \centering
  \caption{Simulation results for the band crossover approach for $h^2=4.00$ and varying $k$ using $10^7$ replicates.}
  \label{tab:mc_band_vary_k}
  \begin{tabular}{ccc}
    \hline
    $k$ & $\widetilde \beta_{k, M}$ & Standard error \\
    \hline
    5 & 0.859650 & 0.000075 \\
    6 & 0.854624 & 0.000069 \\
    7 & 0.851439 & 0.000064 \\
    8 & 0.849514 & 0.000060 \\
    \hline
  \end{tabular}
\end{table}

Observe that Theorem \ref{thm:subgamma_concentration} applies to our proposed approach, as the same sandwiching relation holds. The optimal path length in \eqref{eq:band_exp} is also lower-bounded by the horizontal span between the endpoints and upper-bounded by the identity path. 
\begin{corollary}
    \label{cor:conc_beta_crossover} Fix any $k$, $M$ and $h$. For any $\delta>0$ we have that
    \begin{align}
    & \bP\left( |\widetilde \beta_{k, M}  - \widetilde\beta_k| \geq  \varepsilon \right)  \leq \delta, \nonumber
\end{align}
where $\varepsilon$ is defined in Theorem \ref{thm:conc_beta_chernoff_correct}.
\end{corollary}

Using Corollary \ref{cor:conc_beta_crossover} with $\delta=0.001, M=10^7, h^2=4.00, k=4$, we have $\varepsilon \approx 1.6100\times 10^{-3}$. With $k=8$, we have $\varepsilon \approx 1.5157 \times 10^{-3}$, suggesting these Monte Carlo estimates are tight.

\subsection{Valid upper bound using the band crossover strategy}

In this section, we rigorously evaluate the upper-bound in \eqref{eq:band_exp} to yield a valid improvement over the $0.90380$ bound given by the tuple optimization approach. First, by the triangle inequality, we have
\begin{align}
    \widetilde \beta_k &\leq \frac{1}{kh}\E\Bigg[  \min\Bigg\{ \min_{\pi\in\Pi_k}\sum_{i=1}^k \Bigg\| \begin{pmatrix}
        X_{\pi(i)} -  X_{\pi(i-1)} \\ h^2(U_{\pi(i)} -  U_{\pi(i-1)}) 
    \end{pmatrix} \Bigg\|,  \nonumber \\
    &\quad  2h( R+h-hU_{J^\star}) + \min_{\xi\in\Xi_k}\sum_{i=1}^{k-1} \Bigg\| \begin{pmatrix}
        X_{\pi(i)} -  X_{\pi(i-1)} \\ h^2(U_{\pi(i)} -  U_{\pi(i-1)}) 
    \end{pmatrix} \Bigg\| \Bigg\}  \Bigg], \nonumber \\
    &\stackrel{\Delta}{=} \frac{1}{kh} \E\left[ \widetilde L(X_0, \dots, X_k, U_0, \dots, U_k, R) \right]. \label{eq:band_exp_compact}
\end{align}

Denote $v=(x_0, \dots x_k, u_0, \dots, u_k, r)$ and $g(v)$ as the joint density of $(X_0, \dots, X_k, U_0, \dots, U_k, R).$ Then, \eqref{eq:band_exp_compact} is equivalent to writing

\begin{align}
    \widetilde \beta_k &\leq \frac{1}{kh}\int_\cD \widetilde L(v) g(v) dv, \label{eq:band_exp_int}
\end{align}
where the domain of integration is 
\[ \cD = \{ (x_0, \dots x_k, u_0, \dots, u_k, r): 0 \leq x_0 \leq \dots \leq x_k, 0\leq u_i \leq 1, 0\leq r \}. \]
Suppose we have a region $\cS \subset \cD$ in which the band crossover is preferable to the tuple optimization. Then, we can improve the tuple optimization bound by writing 
\begin{align}
    \widetilde \beta_k &\leq \frac{1}{kh}\int_\cD \widetilde L(v) g(v) dv \nonumber \\
    &= \frac{1}{kh}\int_{\cD \setminus \cS} \widetilde L(v) g(v) dv + \frac{1}{kh}\int_\cS \widetilde L(v) g(v) dv \nonumber \\
    &\leq \frac{1}{kh}\int_{\cD \setminus \cS}  L(v) g(v) dv + \frac{1}{kh}\int_\cS \widetilde L(v) g(v) dv \nonumber \\
    &= \frac{1}{kh}\int_{\cD}  L(v) g(v) dv - \frac{1}{kh} \left( \int_{\cS}  L(v) g(v) dv - \int_\cS \widetilde L(v) g(v) dv \right) \nonumber \\
    &= \widehat \beta_k - \frac{1}{kh} \left( \int_{\cS}  L(v) g(v) dv - \int_\cS \widetilde L(v) g(v) dv \right). \label{eq:beta_tilde_decomp}
\end{align}
The first term in \eqref{eq:beta_tilde_decomp} was upper-bounded by \cite{YuCarlsson2023} using the tuple optimization approach. The second term inside the parentheses quantifies the improvement realized by band crossovers over some region $\cS$. Since \cite{YuCarlsson2023} obtained their best bound of $0.90380$ at $k=4$ and $h^2=3.25$, it follows that 
\begin{align}
    \widetilde \beta_4 &\leq 0.90380 - \frac{1}{4\sqrt{3.25}}\left( \int_{\cS}  L(v) g(v) dv - \int_\cS \widetilde L(v) g(v) dv \right). \label{eq:beta_tilde_4_decomp}
\end{align}
It remains to rigorously lower bound the improvement inside the parentheses in \eqref{eq:beta_tilde_4_decomp} under the same set of parameters. For the rest of the section, we fix $k=4$. 

It is intuitive to see that band crossovers can yield potential savings when one interior vertex of the tuple is located near the top boundary and the remaining vertices are all close to the bottom of the band, as illustrated in Fig. \ref{fig:k4}. 
\begin{figure}[t]
\centering
\includegraphics[width=0.9\linewidth]{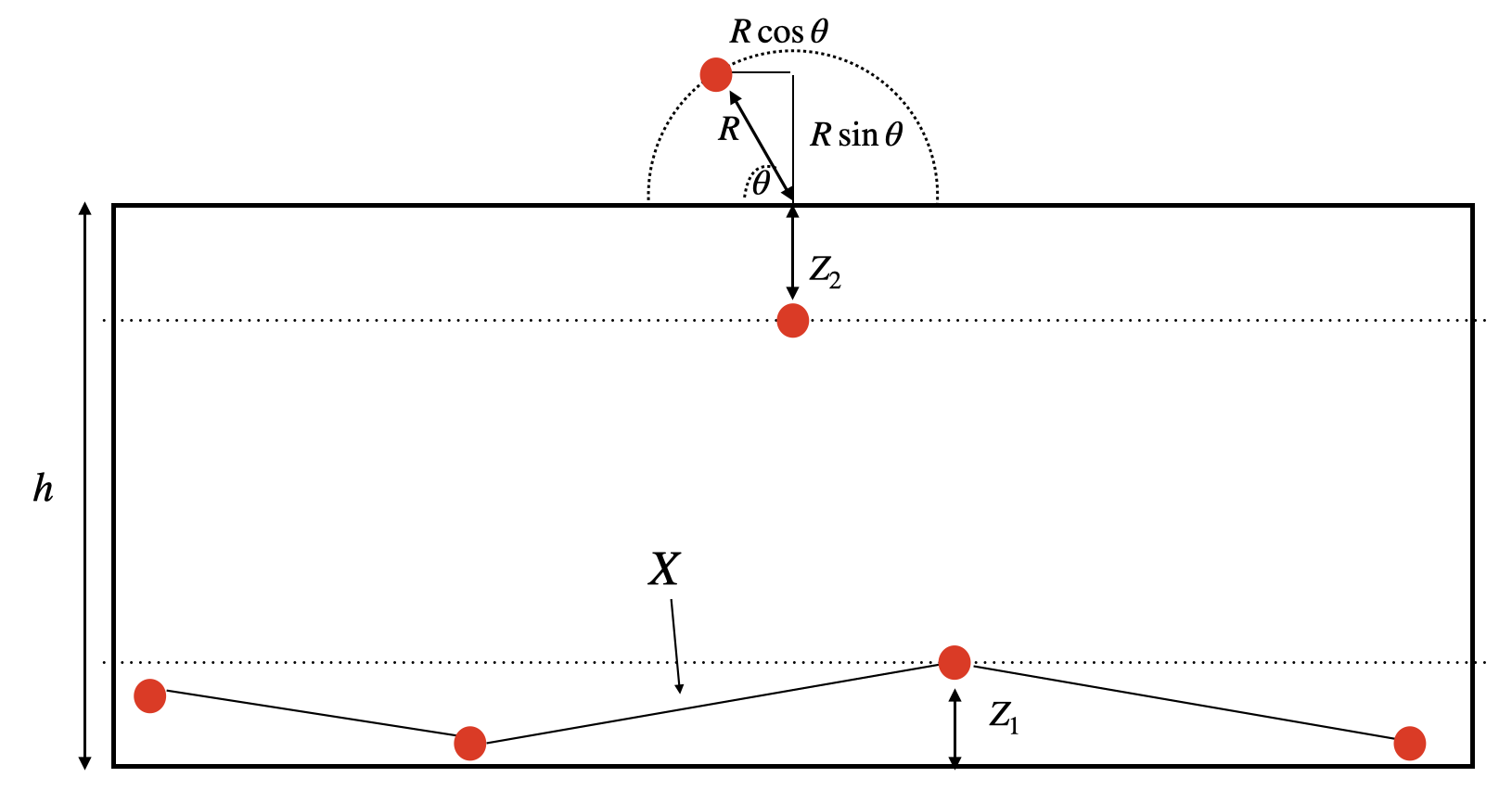}
\caption{An example configuration satisfying $\cS$ for $k=4$}\label{fig:k4}
\end{figure}
Therefore, we define $\cS$ to be the set of such configurations. Let $I$ be the indicator that an interior vertex is the highest vertex among the tuple. Let $Z_1$ be the height of the second highest vertex measured from the bottom of band and let $Z_2$ be the height of the highest vertex measured from the top of the band. Denote $u$ as the highest interior vertex. 
Let $X$ be the optimal path 
through the 4 remaining points excluding $u$. We have, due to the triangle inequality, an upper bound on $\widetilde L(v)$: 
\begin{align}
        \E[\widetilde L(v) | I,  Z_1 = z_1, Z_2=z_2, R=r] &\leq \E[X | I, Z_1=z_1] + 2(z_2+r). \label{eq:crossover_length}
\end{align}

On the other hand, we have a lower bound involving $\E[X | I, Z_1=z_1]$ on $L(v)$ based on the following lemma.
\begin{lemma}
    \begin{align}
        \E[L(v) | I, Z_1 = z_1, Z_2=z_2, R=r] &\geq \E[X | I, Z_1=z_1] + g(h-z_1-z_2), \label{eq:tuple_length}
    \end{align}
    where $g(c) = \sqrt{2}c - \frac{\sqrt{2}-1}{\sqrt{2}} \cdot \frac{3.75}{h}$
\end{lemma}
\begin{proof}
    For brevity, we condition on $I$, $Z_1 = z_1, Z_2=z_2, R=r$ throughout this proof. 
    Under the optimal path $L$ through all $5$ points of the tuple, we have a $2$-path $(v, u, w)$ for some vertices $v, w$ in the bottom region of the band. Deleting these two edges and adding the edge $(v, w)$ yields a path through the bottom $4$ points, which is an upper bound on the length of the optimal path through the bottom $4$ points. Since $X$ is defined to be optimal, it follows that 
    \begin{align}
        X \cdot I &\leq \left(L - \|(v, u)\| - \|(u, w)\| + \|(v, w)\|\right) \cdot I. \nonumber
    \end{align}
    Denote $\Delta:= \left( \|(v, u)\| + \|(u, w)\| - \|(v, w)\| \right) \cdot I $. 
    Due to symmetry, there are $3$ possible configurations of $u, v, w$, as shown in Fig. \ref{fig:triangle}. 

\begin{figure}[t]
\centering
\includegraphics[width=0.95\linewidth]{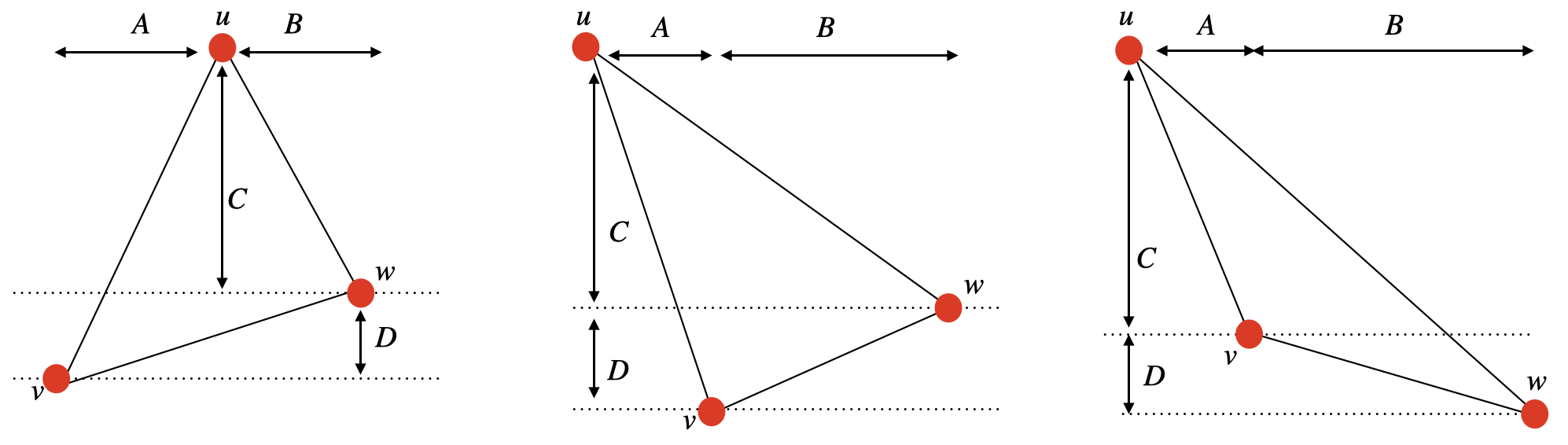}
\caption{Possible configurations of $u, v, w$.}\label{fig:triangle}
\end{figure}

\paragraph{Case 1} 
\begin{align*}
    \Delta 
    &= \sqrt{A^2+(C+D)^2} + \sqrt{C^2 + B^2} - \sqrt{D^2 + (A+B)^2} 
\end{align*}
\paragraph{Case 2} 
\begin{align*}
    \Delta 
    &= \sqrt{A^2+(C+D)^2} + \sqrt{C^2 + (A+B)^2} - \sqrt{D^2 + B^2}
\end{align*}
\paragraph{Case 3} 
\begin{align*}
    \Delta 
    &= \sqrt{A^2+C^2} + \sqrt{(A+B)^2+(C+D)^2} - \sqrt{D^2 + B^2} 
\end{align*}
Since $A$ and $B$ are interchangeable, Case 1 lower-bounds the other two cases. Differentiating the RHS in Case 1 with respect to $D$ yields
\[ \frac{C+D}{\sqrt{A^2+(C+D)^2}} - \frac{D}{\sqrt{D^2 + (A+B)^2}} \geq 0, \]
so the bound is increasing in $D$. Setting $D=0$ and using the bound $\sqrt{x^2+y^2} \geq (x+y) / \sqrt{2}$ yields a uniform lower bound:
\begin{align*}
    \Delta &\geq \sqrt{A^2+C^2} + \sqrt{C^2 + B^2} - (A+B) \\
    &\geq \frac{1}{\sqrt{2}}(A + B + 2C) - (A+B) \\
    &= \sqrt{2}C - \frac{\sqrt{2}-1}{\sqrt{2}}(A+B).
\end{align*}
We note that $h(A+B)$ is stochastically dominated by the sum of the three largest exponential random variables with parameter $h$ among four. The expectation of the quantity is the same as the expectation of $Erlang(4,h)$ minus the expectation of the minimum of 4 independent exponential random variables with parameter $h$. Therefore, 
\[ \E[A+B] \leq \frac{1}{h}\left( 4 - \frac{1}{4} \right) = \frac{3.75}{h}. \]
Since $C\geq  h - Z_1 - Z_2$ due to the conditioning event, taking the expectation of $\Delta$ yields
\begin{align*}
    \E\Delta &\geq \sqrt{2}(h-z_1-z_2) - \frac{\sqrt{2}-1}{\sqrt{2}} \cdot \frac{3.75}{h},
\end{align*}
which concludes the proof.

\end{proof}

Set $\alpha= \frac{1-\sqrt{2}}{\sqrt{2}} \cdot \frac{3.75}{h}$. Our procedure is guaranteed to yield an improvement whenever the path length using the band crossover in \eqref{eq:crossover_length} is shorter than that using the tuple optimization in \eqref{eq:tuple_length}, which is implied when $2(z_2+r) < \sqrt{2}(h-z_1-z_2) + \alpha$.
Denote $I_1, I_2, I_3$ as the indicators that the highest point is the first, second, and third interior vertex, respectively. Since we can conduct the improvement procedure when the highest vertex is one of these three vertices, the expected improvement is 
\begin{align}
    &\E\left[\max\left\{0, \sqrt{2}(h-z_1-z_2) - 2(z_2+r) + \alpha \right\}(I_1 + I_2 + I_3)\right] \nonumber \\
    &= 3\E\left[\max\left\{0, \sqrt{2}(h-z_1-z_2) - 2(z_2+r) + \alpha\right\}I_1\right] \nonumber\\
    &= 3\E\left[\max\left\{0, \sqrt{2}(h-z_1-z_2) - 2(z_2+r) + \alpha \right\}\mid I_1\right]\bP(I_1=1) \nonumber\\
    &= \frac{3}{5}\E\left[\max\left\{0, \sqrt{2}(h-z_1-z_2) - 2(z_2+r) + \alpha \right\}\}\mid I_1\right] \nonumber\\
    &= \frac{3}{5}\E\left[\max\left\{0, \sqrt{2}(h-z_1-z_2) - 2(z_2+r) + \alpha \right\}\right]. \label{eq:E_max_intermediate}
\end{align}
The conditioning is dropped in \eqref{eq:E_max_intermediate} because the joint distribution of $(Z_1, Z_2, R)$ is the same as the joint conditional distribution of $(Z_1, Z_2, R) | \{I_1 = 1\}$. 

For $z_2 \leq h$, the CDF of $Z_2$ is given by
    \[P(Z_2 \leq z_2) = 1 - \left(\frac{h-z_2}{h}\right)^5,\]
    so the PDF of $Z_2$ is $f_{Z_2}(z_2) = 5(h-z_2)^4 / h^5$. The CDF of $Z_1$ given that $Z_2 = z_2$ is, for $0 \leq z_1 \leq h-z_2$,
    \[P(Z_1 \leq z_1 | Z_2 = z_2) =  \left(\frac{z_1}{h - z_2} \right)^4,\]
    and the conditional density of $Z_1$ given $Z_2 = z_2$ is 
    \[f_{Z_1|Z_2}(z_1 |z_2) =  4 \left(\frac{1}{h - z_2} \right)^4 \cdot z_1^3.\]
    Finally, $R$ is independent of $Z_1, Z_2$. Therefore, the joint density of $(Z_1, Z_2, R)$ evaluated at $(z_1, z_2, r)$ is given by
    \begin{align}
    f_{Z_2}(z_2) \cdot f_{Z_1 | Z_2}(z_1 | z_2) \cdot f_R(r) &=   \frac{20}{h^5}  \mathbbm{1}\{z_1 + z_2 \leq h, z_1, z_2\geq0 \} z_1^3 \pi r e^{-\frac{\pi}{2} r^2} \label{eq:joint_density}
    \end{align}

Denote $(x)_+ = \max(0, x).$ Plugging \eqref{eq:joint_density} into \eqref{eq:E_max_intermediate} and using change-of-variables, our expected improvement is at least
\begin{align}
    &\frac{3\cdot20}{5h^5} \int_0^h \int_0^h \int_0^\infty \left(\sqrt{2}(h-z_1-z_2) - 2(z_2+r) + \alpha \right)_+ \nonumber \\
    &\quad \mathbbm{1}\{z_1 + z_2 \leq h \} z_1^3 \pi r e^{-\frac{\pi}{2} r^2} dr dz_2 dz_1 \nonumber \\
    &= 12\int_0^1 \int_0^{1-u_1} \int_0^\infty \left(\sqrt{2}(h-hu_1-hu_2) - 2(hu_2+r) + \alpha \right)_+ \nonumber \\
    &\qquad  u_1^3 \pi r e^{-\frac{\pi}{2} r^2} dr du_2 du_1. \label{eq:3d_int}
\end{align}

To evaluate the inner integral with respect to $r$, we use the following identity.
\begin{equation}
    F(a) := \int_0^a (a-r)\pi r e^{-\pi r^2 / 2} dr = a - \frac{1}{\sqrt{2}} \erf\left(\sqrt{\frac{\pi}{2}}a\right),  \label{eq:rayleigh_identity}
\end{equation}
where $\erf(\cdot)$ is the error function.
Let
\[ a(u_1, u_2) = \frac{h}{2} \left( \sqrt{2} - \sqrt{2}u_1 - (\sqrt{2}+2)u_2 \right) + \frac{\alpha}{2}. \]
Then, via \eqref{eq:rayleigh_identity}, \eqref{eq:3d_int} reduces to 
\begin{align}
    24 \int_0^1 \int_0^{1-u_1} u_1^3 F((a(u_1, u_2)_+) du_2 du_1. \label{eq:2d_int}
\end{align}

To simplify \eqref{eq:2d_int} further to only with respect to $u_1$, re-write $a(u_1, u_2) = A(u_1) - Bu_2$, where
\begin{align}
    A(u_1) &=  \frac{\sqrt{2}h(1-u_1) + \alpha}{2}  \nonumber \\
    B &= \frac{h(\sqrt{2}+2)}{2}. \nonumber
\end{align}
Since the integrand is nonnegative, the inner integral bounds reduce to 
\begin{align}
    &24 \int_0^1 u_1^3 \int_0^{1-u_1}  F((A(u_1)-Bu_2)_+) du_2 du_1. \nonumber
\end{align}
With $u_1$ fixed, we then conduct a change-of-variables with $a=A(u_1)-Bu_2$. 
Then, the inner integral becomes 
\begin{align}
 \int_0^{1-u_1}  F((A(u_1)-Bu_2)_+) du_2  &= - \frac{1}{B} \int_{A(u_1)}^{A(u_1)-B(1-u_1)}  F((a)_+) da \nonumber \\
 &=\frac{1}{B} \int_{A(u_1)-B(1-u_1)}^{A(u_1)} F((a)_+) da. \label{eq:F_simplified_pre}
\end{align}
Observe that 
\begin{align}
    A(u_1) - B(1-u_1) &= \frac{\sqrt{2}h(1-u_1) + \alpha}{2} - \frac{h(\sqrt{2}+2)}{2} (1-u_1) \nonumber \\
    &= \frac{\alpha}{2} - h(1-u_1) < 0~\text{for}~u_1\in[0, 1].
\end{align}
Therefore, we can truncate the lower integral bound and also modify the upper bound to simplify \eqref{eq:F_simplified_pre} further to 
\begin{equation}
    \frac{1}{B} \int_{0}^{(A(u_1))_+} F(a) da. \label{eq:F_simplified}
\end{equation}

Define $H(a) = \int_0^a F(t)dt$. Use the identity
\[ \int \erf(x)dx = x\erf(x) + \frac{e^{-x^2}}{\sqrt{\pi}} + C \]
to obtain 
\begin{align}
    H(a) &= \int_0^a \left[t - \frac{1}{\sqrt{2}}\erf\left(\sqrt{\frac{\pi}{2}}t\right) \right] dt \nonumber \\
    &= \int_0^{\sqrt{\frac{\pi}{2}}a} \left[\sqrt{\frac{2}{\pi}}x - \frac{1}{\sqrt{2}}\erf\left(x\right) \right] \sqrt{\frac{2}{\pi}} dx \nonumber \\
    &= \frac{2}{\pi} \cdot \frac{1}{2}x^2 \Big|_0^{\sqrt{\frac{\pi}{2}}a} - \frac{1}{\sqrt{\pi}}x\erf(x)\Big|_0^{\sqrt{\frac{\pi}{2}}a} - \frac{1}{\pi}e^{-x^2} \Big|_0^{\sqrt{\frac{\pi}{2}}a}\nonumber \\
    &=\frac{a^2}{2} - \frac{a}{\sqrt{2}}\erf\left(a\sqrt{\frac{\pi}{2}}\right) - \frac{1}{\pi}e^{-\frac{\pi}{2}a^2} + \frac{1}{\pi}. \label{eq:Ha_expression}
\end{align}
Combining \eqref{eq:F_simplified} and \eqref{eq:Ha_expression} reduces \eqref{eq:2d_int} to 
\begin{align}
    \eta &:= \frac{24}{B}\int_0^1 u_1^3 \int_{0}^{(A(u_1))_+} F(a) da du_1 \nonumber  \\
    &=\frac{24}{B}\int_0^1 u_1^3 H((A(u_1))_+)   du_1 \nonumber \\
    &= \frac{24}{B} \int_0^{u_1^\star} u_1^3 H(A(u_1))du_1, \label{eq:I_1d}
\end{align}
where $u_1^\star = \min\{1, 1 + \frac{\alpha}{\sqrt{2}h}\} = 1 + \frac{\alpha}{\sqrt{2}h}$.
Observe that $H'(a) = F(a) \geq 0$ for $a\geq 0$ so $H$ is increasing. By the chain rule,
\begin{equation}
    \frac{d}{du_1} H(A(u_1)) = H'(A(u_1)) A'(u_1)) = H'(A(u_1)) \cdot \frac{-h}{\sqrt{2}} \leq 0. \nonumber
\end{equation}
As a result, $H(A(u_1))$ is decreasing in $u_1$ on $[0, u_1^\star].$ Thus, we can rigorously lower-bound \eqref{eq:I_1d} using the following Riemann sum scheme. We fix a partition of $[0, u_1^\star]$:
\[ 0=u_1 < u_1 < \dots u_N=u_1^\star. \]
Since $u_1^3$ is increasing in $u_1$, $\eta$ is lower bounded by
\begin{align}
    &\eta \geq \frac{24}{B}\sum_{i=1}^N (u_i - u_{i-1})u_{i-1}^3 H(A(u_i)) \label{eq:I_numerical}
\end{align}
which plugs into \eqref{eq:beta_tilde_4_decomp} as
\begin{align}
    \widetilde \beta_4 &\leq 0.90380 - \frac{1}{4\sqrt{3.25}}\left(0.000987 \right) \leq 0.90367,
\end{align}
improving the bound of \cite{YuCarlsson2023} by $0.00013$.

\section{Conclusion}

The computer-aided tuple-optimization framework of \cite{YuCarlsson2023} is the best-to-date upper bound on the Euclidean TSP constant, and our results clarify both its current promise and its likely limitations. On the one hand, Monte Carlo experiments indicate that, even with perfect identification of the optimal permutation on each region of the integration domain, band-restricted tuple optimization appears to plateau around $0.86-0.88$ (depending on $k$), so that further computational refinement of the numerical analysis can plausibly chip away at the constant, but is unlikely to drive the bound close to the conjectured value $\beta\approx 0.71$. On the other hand, our band-crossover heuristic delivers a uniform improvement over tuple optimization for every $(k,h)$ by exploiting near-boundary geometry, and we rigorously certify a numerical improvement of the best known upper bound, reducing $0.90380$ to $0.90367$.

While the actual numerical improvement is modest, it should not be interpreted as a weakness of the underlying heuristic. Rather, it reflects the looseness introduced by the analytic inequalities required to make the upper-bounding argument rigorous — most notably, repeated applications of triangle-inequality type bounds that “pay” extra length exactly in the regimes where cross-band connections are most beneficial. In contrast, the Monte Carlo results show that allowing departures from the strict band constraint produces substantial empirical gains (on the order of $10^{-2}$ for $k=4$ and approaching $0.85$ for $k=8$), suggesting great promise in designing certifiable tour-construction heuristics that more fundamentally move beyond within-band traversal for further lowering of the upper bound for $\beta$.

\textbf{Acknowledgements} J. G. and C. G. were supported in part by NSF CAREER Award CCF-2440539. C. G. was supported in part by a Dr. John N. Nicholson Fellowship, Northwestern University.

\bibliographystyle{apalike}
\bibliography{refs.bib}

@article{GAUDIO202067,
title = {An improved lower bound for the Traveling Salesman constant},
journal = {Operations Research Letters},
volume = {48},
number = {1},
pages = {67-70},
year = {2020},
issn = {0167-6377},
doi = {https://doi.org/10.1016/j.orl.2019.11.007},
url = {https://www.sciencedirect.com/science/article/pii/S0167637719304924},
author = {Julia Gaudio and Patrick Jaillet}
}

@article{YuCarlsson2023,
author = {Carlsson, John Gunnar and Yu, Julien},
title = {A New Upper Bound for the {Euclidean} {TSP} Constant},
journal = {INFORMS Journal on Computing},
volume = {},
number = {},
pages = {},
year = {2024},
doi = {10.1287/ijoc.2024.0538},

URL = { 
    
        https://doi.org/10.1287/ijoc.2024.0538
    
    

},
eprint = { 
    
        https://doi.org/10.1287/ijoc.2024.0538
    
    

}
}

@article{BHH,
  title={The shortest path through many points},
  author={Jillian E. Beardwood and John H. Halton and J. M. Hammersley},
  journal={Mathematical Proceedings of the Cambridge Philosophical Society},
  year={1959},
  volume={55},
  pages={299 - 327},
  url={https://api.semanticscholar.org/CorpusID:122062088}
}

@article{Steinerberger,
 ISSN = {00018678},
 URL = {http://www.jstor.org/stable/43563459},
 author = {Stefan Steinerberger},
 journal = {Advances in Applied Probability},
 number = {1},
 pages = {27--36},
 publisher = {Applied Probability Trust},
 title = {NEW BOUNDS FOR THE TRAVELING SALESMAN CONSTANT},
 urldate = {2025-03-20},
 volume = {47},
 year = {2015}
}

@book{Applegate2006,
  author    = {David L. Applegate and Robert E. Bixby and V{\'{a}}clav Chv{\'{a}}tal and William J. Cook},
  title     = {The Traveling Salesman Problem: A Computational Study},
  publisher = {Princeton University Press},
  address   = {Princeton},
  year      = {2006},
  isbn      = {9780691129938}
}

@misc{concorde_tsp_solver,
  author       = {Applegate, David and Bixby, Robert E. and Chv{\'a}tal, Va{\v{s}}ek and Cook, William J.},
  title        = {Concorde {TSP} Solver},
  howpublished = {\url{https://www.math.uwaterloo.ca/tsp/concorde.html}},
  note         = {Version 03.12.19 (Dec 19, 2003). Accessed: 2026-01-05},
  year         = {2003}
}

@article{stein1977,
 ISSN = {00411655, 15265447},
 URL = {http://www.jstor.org/stable/25767916},
 author = {David M. Stein},
 journal = {Transportation Science},
 number = {3},
 pages = {232--249},
 publisher = {INFORMS},
 title = {Scheduling Dial-a-Ride Transportation Systems},
 urldate = {2026-01-05},
 volume = {12},
 year = {1978}
}

@article{ONG1989231,
title = {Asymptotic expected performance of some {TSP} heuristics: An empirical evaluation},
journal = {European Journal of Operational Research},
volume = {43},
number = {2},
pages = {231-238},
year = {1989},
issn = {0377-2217},
doi = {https://doi.org/10.1016/0377-2217(89)90217-8},
url = {https://www.sciencedirect.com/science/article/pii/0377221789902178},
author = {H.L. Ong and H.C. Huang}
}

@article{Fiechter1994APT,
  title={A Parallel Tabu Search Algorithm for Large Traveling Salesman Problems},
  author={Claude-Nicolas Fiechter},
  journal={Discret. Appl. Math.},
  year={1994},
  volume={51},
  pages={243-267},
  url={https://api.semanticscholar.org/CorpusID:41829903}
}

@article{LeeChoi1994,
  title = {Optimization by multicanonical annealing and the traveling salesman problem},
  author = {Lee, Jooyoung and Choi, M. Y.},
  journal = {Phys. Rev. E},
  volume = {50},
  issue = {2},
  pages = {R651--R654},
  numpages = {0},
  year = {1994},
  month = {Aug},
  publisher = {American Physical Society},
  doi = {10.1103/PhysRevE.50.R651},
  url = {https://link.aps.org/doi/10.1103/PhysRevE.50.R651}
}

@inproceedings{Johnson1996AsymptoticEA,
  title={Asymptotic experimental analysis for the {Held}-{Karp} traveling salesman bound},
  author={David S. Johnson and Lyle A. McGeoch and Edward E. Rothberg},
  booktitle={ACM-SIAM Symposium on Discrete Algorithms},
  year={1996},
  url={https://api.semanticscholar.org/CorpusID:1766406}
}

@article{PercusMartin,
  title = {Finite Size and Dimensional Dependence in the {Euclidean} Traveling Salesman Problem},
  author = {Percus, Allon G. and Martin, Olivier C.},
  journal = {Phys. Rev. Lett.},
  volume = {76},
  issue = {8},
  pages = {1188--1191},
  numpages = {0},
  year = {1996},
  month = {Feb},
  publisher = {American Physical Society},
  doi = {10.1103/PhysRevLett.76.1188},
  url = {https://link.aps.org/doi/10.1103/PhysRevLett.76.1188}
}

@article{VALENZUELA1997157,
title = {Estimating the {Held}-{Karp} lower bound for the geometric {TSP}},
journal = {European Journal of Operational Research},
volume = {102},
number = {1},
pages = {157-175},
year = {1997},
issn = {0377-2217},
doi = {https://doi.org/10.1016/S0377-2217(96)00214-7},
url = {https://www.sciencedirect.com/science/article/pii/S0377221796002147},
author = {Christine L. Valenzuela and Antonia J. Jones}
}

@article{JacobsenReadSaleur,
  title = {Traveling Salesman Problem, Conformal Invariance, and Dense Polymers},
  author = {Jacobsen, J. L. and Read, N. and Saleur, H.},
  journal = {Phys. Rev. Lett.},
  volume = {93},
  issue = {3},
  pages = {038701},
  numpages = {4},
  year = {2004},
  month = {Jul},
  publisher = {American Physical Society},
  doi = {10.1103/PhysRevLett.93.038701},
  url = {https://link.aps.org/doi/10.1103/PhysRevLett.93.038701}
}

@book{concentration_inequalities,
    author = {Boucheron, Stéphane and Lugosi, Gábor and Massart, Pascal},
    title = {Concentration Inequalities: A Nonasymptotic Theory of Independence},
    publisher = {Oxford University Press},
    year = {2013},
    month = {02},
    isbn = {9780199535255},
    doi = {10.1093/acprof:oso/9780199535255.001.0001},
    url = {https://doi.org/10.1093/acprof:oso/9780199535255.001.0001},
}

\end{document}